\documentclass[letterpaper, 10 pt, conference]{ieeeconf}  
\IEEEoverridecommandlockouts
\usepackage{balance}
\usepackage{color}
\usepackage{caption}
\usepackage{enumerate}
\usepackage{cite}
\usepackage{empheq}
\usepackage{mathrsfs}

\usepackage{mathtools}

\usepackage{tikz}
\usepackage{circuitikz}
\usepackage{flushend}

\usepackage[font=small,labelfont=bf]{caption}

\usepackage{subfig}


\makeatletter
\pgfcircdeclarebipole{}{\ctikzvalof{bipoles/interr/height 2}}{spst}{\ctikzvalof{bipoles/interr/height}}{\ctikzvalof{bipoles/interr/width}}{

    \pgfsetlinewidth{\pgfkeysvalueof{/tikz/circuitikz/bipoles/thickness}\pgfstartlinewidth}

    \pgfpathmoveto{\pgfpoint{\pgf@circ@res@left}{0pt}}
    \pgfpathlineto{\pgfpoint{.6\pgf@circ@res@right}{\pgf@circ@res@up}}
    \pgfusepath{draw}   
}

\def\pgf@circ@spst@path#1{\pgf@circ@bipole@path{spst}{#1}}
\tikzset{switch/.style = {\circuitikzbasekey, /tikz/to path=\pgf@circ@spst@path, l=#1}}
\tikzset{spst/.style = {switch = #1}}
\makeatother

\makeatletter
\let\proof\@undefined                        
\let\endproof\@undefined                  
\makeatother
\usepackage{graphicx,amssymb,amstext,amsmath,amsthm}

\usepackage[bookmarks=true]{hyperref}
\usepackage{algorithm,algorithmicx,algpseudocode}
\algnewcommand{\algorithmicgoto}{\textbf{go to}}%
\algnewcommand{\Goto}[1]{\algorithmicgoto~\ref{#1}}%
\algnewcommand{\LineComment}[1]{\Statex \(\triangleright\) #1}
\algnewcommand{\LineCommentN}[1]{\Statex \hspace{1cm}\(\triangleright\) #1}

\usepackage{multirow}
\usepackage{stfloats}

\newtheorem{prop}{Proposition} 

\newtheorem{thm}{Theorem}
	\newtheorem{assumption}{Assumption}
\newtheorem{lem}{Lemma}
\newtheorem{defn}{Definition}

\newtheorem{problem}{Problem}

\usepackage{setspace}
\let\oldbibliography\thebibliography
\renewcommand{\thebibliography}[1]{%
  \oldbibliography{#1}%
}

\newcommand{\tdiag}{\textstyle{\mathrm{diag}}}

\newcommand{\yong}[1]{{\color{black} #1}}
\newcommand{\moh}[1]{{\color{black} #1}}

\newcommand{\fa}[1]{{\color{black} #1}}
\newcommand{\mk}[1]{{\color{black} #1}}

\begin{document}

\title{\LARGE \bf Optimal Feedback Stabilizing Control of Bounded Jacobian Discrete-Time Systems via Interval Observers} 

\author{%
Mohammad Khajenejad\\
\thanks{
M. Khajenejad is with the Departments of Mechanical Engineering and Electrical and Computer Engineering, University of Tulsa, Tulsa, OK, USA. (e-mail: mohammad-khajenejad@utulsa.edu.)}
}

\maketitle
\thispagestyle{empty}
\pagestyle{empty}

\begin{abstract}
This paper addresses optimal feedback stabilizing control for bounded Jacobian nonlinear discrete-time (DT) systems with nonlinear observations, affected by state and process noise. Instead of directly stabilizing the uncertain system, we propose stabilizing a higher-dimensional interval observer whose states enclose the true system states. Our nonlinear control approach introduces additional flexibility compared to linear methods, compensating for system nonlinearities and allowing potentially tighter closed-loop intervals. We also establish a separation principle, enabling independent design of observer and control gains, and derive tractable linear matrix inequalities, resulting in a stable closed-loop system.
 \end{abstract}

\section{Introduction} 

In many real-world applications, solving estimation and control problems is complicated by the presence of significant nonlinearities and uncertainties in process models. For instance, parameter identification and system modeling in fields like systems biology and neuroscience are highly complex due to inherent nonlinear dynamics \cite{izhikevich2007dynamical}. Similarly, collaborative robotics, where multiple agents—including humans—interact, introduces considerable uncertainty, making accurate modeling and control challenging \cite{meurer2005control}. In such contexts, these uncertainties cannot be ignored and must be explicitly accounted for in the design process, necessitating adjustments to both control and estimation objectives. On the other hand, traditional control techniques often rely on precise system models, which can be difficult to obtain or maintain in practical applications due to noise, modeling errors, and external disturbances. This makes the design of estimation-based control approaches for nonlinear systems extremely important, particularly in cases where the system faces bounded uncertainties with unknown distributions or stochastic characteristics—rendering typical Kalman filter-based approaches inapplicable.

In such cases set-valued, and in particular interval-valued observers \cite{mazenc2011interval,wang2015interval,chebotarev2015interval,khajenejad2024distributed} have been leveraged to calculate a set (interval) of admissible values for the state at each instant of time using input-output information and the bounds on
uncertainty \cite{moisan2007near,raissi2010interval,tahir2021synthesis,khajenejad2021intervalACC,khajenejad2020simultaneousCDC}. The interval length is minimized by directly relying on monotone systems theory \cite{farina2000positive} to obtain cooperative observer error dynamics, by leveraging interval arithmetic or M\"{u}ller's theorem-based approaches \cite{kieffer2006guaranteed}, or by designing observer gains via solving linear, semidefinite or mixed integer programs~\cite{efimov2013interval,9790824,9803272}, making the gains proportional to the magnitudes of the model uncertainties.

In the context of control design, and in particular for linear continuous-time (CT) systems, the work in \cite{wang2020intervalcontrol} proposed an interval observer-based state feedback controller for switched linear systems, while \cite{zhang2024semi} proposed a framework for semi-global interval observer-based robust control of linear time-invariant (LTI) systems subject to input saturation. 

Moreover, concerning nonlinear settings, the work in \cite{efimov2013control} introduced a dynamic output feedback approach based on an interval state observer design for a class of nonlinear continuous-time (CT) systems,  which was restricted to linear observations and did not consider measurement noise. Similarly, the research in \cite{he2018control} proposed a Luenberger-like interval observer methodology for nonlinear CT systems with linear observations and without taking to account sate and measurement noise. Recently, in the research in \cite{efimov2024interval} a nonlinear feedback control strategy was proposed for a class of nonlinear CT systems using interval estimates, where sufficient conditions for stability of estimation and regulation errors were formulated using linear matrix inequalities (LMI)s. However, non of these works considered nonlinear discrete-time (DT) systems with nonlinear observations and state and measurement noise. 

{\emph{{Contributions}.}} To bridge this gap, this paper considers optimal stabilizing feedback control of bonded Jacobian nonlinear DT systems with nonlinear observations which are subject to state and process noise. We show that instead of stabilizing the true uncertain system, we can stabilize a higher dimensional certain \emph{framer system}, i.e., an interval observer, whose states contain the states of the actual state from below and above.  our nonlinear control approach benefits from an additional degree of freedom  compared to the linear approach in \cite{efimov2013control}. This extra gain compensates for the potential dynamic instabilities that are due to the system nonlinearities and expands the feasible set of admissible control signals, potentially resulting in tighter closed-loop intervals. Moreover, we provide a \emph{separation principle} showing that the observer and control gains can be designed separately. This paves the way for the ability to obtain tractable LMIs to design control gains.
  
 \section{Preliminaries}
 
 {\emph{{Notation}.}} $\mathbb{R}^n,\mathbb{R}^{n  \times p},\mathbb{D}_n,\mathbb{N},\mathbb{N}_n$ denote the $n$-dimensional Euclidean space and the sets of $n$ by $p$ matrices, $n$ by $n$ diagonal matrices, natural numbers and natural numbers up to $n$, respectively. 
 For $M \in \mathbb{R}^{n \times p}$, $M_{i,j}$ denotes $M$'s entry in the $i$'th row and the $j$'th column, $M^\oplus\triangleq \max(M,\mathbf{0}_{n \times p})$, $M^\ominus=M^\oplus-M$ and $|M|\triangleq M^\oplus+M^\ominus$, where $\mathbf{0}_{n \times p}$ is the zero matrix in $\mathbb{R}^{n \times p}$ and $\mathbf{0}_{n}$ is the zero vector in $\mathbb{R}^n$, \yong{while $\textstyle{\mathrm{sgn}}(M) \in \mathbb{R}^{n \times p}$ is the element-wise sign of $M$ with $\textstyle{\mathrm{sgn}}(M_{i,j})=1$ if $M_{i,j} \geq 0$ and $\textstyle{\mathrm{sgn}}(M_{i,j})=-1$, otherwise.} 
 Furthermore, if $p=n$, 
 \mk{$M \succ \mathbf{0}_{n \times n}$ and $M \prec \mathbf{0}_{n \times n}$ (or $M \succeq \mathbf{0}_{n \times n}$ and $M \preceq \mathbf{0}_{n \times n}$) denote that $M$ is positive and negative (semi-)definite, respectively}. Finally, $\tdiag(A_1,\dots, A_n)$ denotes a block diagonal matrix with $A_1, \dots,A_n$ being its diagonal block matrix entries, while $\tdiag_n(A)\triangleq \tdiag(\underbrace{A,\dots,A}_{n \ \text{times}})$. Finally $I_n$ denotes the identity matrix in $\mathbb{R}^n$. 

Next, we introduce some useful definitions and results.
\begin{defn}[Interval, Maximal and Minimal Elements, Interval Width]\label{defn:interval}
{An (multi-dimensional) interval {$\mathcal{I} \triangleq [\underline{s},\overline{s}]  \subset 
\mathbb{R}^n$} is the set of all real vectors $x \in \mathbb{R}^n$ that satisfies $\underline{s} \le x \le \overline{s}$, where $\underline{s}$, $\overline{s}$ and $\|\overline{s}-\underline{s}\|\mk{_{\infty}\triangleq \max_{i \in \{1,\cdots,n\}}s_i}$ are called minimal vector, maximal vector and \mk{interval} width of $\mathcal{I}$, respectively}. \yong{An interval matrix can be defined similarly.} 
\end{defn}
\begin{prop}\cite[Lemma 1]{efimov2013interval}\label{prop:bounding}
Let $A \in \mathbb{R}^{n \times p}$ and $\underline{x} \leq x \leq \overline{x} \in \mathbb{R}^n$. Then\moh{,} $A^+\underline{x}-A^{-}\overline{x} \leq Ax \leq A^+\overline{x}-A^{-}\underline{x}$. As a corollary, if $A$ is non-negative, $A\underline{x} \leq Ax \leq A\overline{x}$. 
\end{prop}

\begin{defn}[Jacobian Sign-Stability] \label{defn:JSS}
A mapping $f :\mathcal{X} \subset \mathbb{R}^{n} \to  \mathbb{R}^{p}$ is (generalized) Jacobian sign-stable (JSS), if its (generalized) Jacobian matrix entries \fa{do} not change signs on its domain, i.e., if either of the following hold: 
\begin{align*}
&\forall x \in \mathcal{X}, \forall i \in \mathbb{N}_p,\forall j \in \mathbb{N}_n , J_f(x)_{i,j} \geq 0 \ \text{(positive JSS)}\\  
&\forall x \in \mathcal{X}, \forall i \in \mathbb{N}_p,\forall j \in \mathbb{N}_n , J_f(x)_{i,j} \leq 0 \  \text{(negative JSS)},
\end{align*} 
where $J_f(x)$ denotes the Jacobian matrix of $f$ at $x \in \mathcal{X}$. 
\end{defn}
\begin{prop}[Mixed-Monotone Decomposition]
\cite[Proposition 2]{9867741}\label{prop:JSS_decomp}
Let $f :\mathcal{X} \subset \mathbb{R}^{n} \to  \mathbb{R}^{p}$ and suppose $\forall x \in \mathcal{X}, J_f(x) \in [\underline{J}_f,\overline{J}_f]$, where $\underline{J}_f,\overline{J}_f$ are known matrices in $\mathbb{R}^{p \times n}$. Then, $f$ can be decomposed into a (remainder) affine mapping $Hx$ and a JSS mapping $\mu (\cdot)$, in an additive form: 
\begin{align}\label{eq:JSS_decomp}
\forall x \in \mathcal{X},f(x)=\mu(x)+Hx,
\end{align}
 where $H$ is a matrix in $\mathbb{R}^{p \times n}$, that satisfies the following 
 \begin{align}\label{eq:H_decomp}
 \forall (i,j) \in \mathbb{N}_p \times \mathbb{N}_n, H_{i,j}=(\overline{J}_f)_{i,j} \ \lor \ H_{i,j}=(\underline{J}_f)_{i,j}.    
 \end{align}
\end{prop}
\begin{prop}[Tight and Tractable Decomposition Functions for JSS Mappings {\cite[Proposition 4 \& Lemma 3]{9867741}}]\label{prop:tight_decomp}
Let $\mu:{\mathcal{Z}} \subset \mathbb{R}^{n_z} \to \mathbb{R}^p$ be a JSS mapping on its domain. Then, it admits a tight decomposition function for each $\mu_i,\ i \in \mathbb{N}_p$ as follows: 
\begin{align}\label{eq:JJ_decomp}
\mu_{d,i}({z}_1,{z}_2)\hspace{-.1cm}=\hspace{-.1cm}\mu_i(D^i{z}_1\hspace{-.1cm}+\hspace{-.1cm}(I_{n_z}\hspace{-.1cm}-\hspace{-.1cm}D^i){z}_2), 
\end{align}
{for any ordered ${z_1, z_2 \in \mathcal{Z}}$}, where $D^i$ 
is a binary diagonal matrix determined by which vertex of the interval {$[{z}_2,{z}_1]$ {or $[z_1,z_2]$}} that maximizes {(if $z_2 \leq z_1$) or minimizes {(if $z_2 > z_1$)}} the function 
$\mu_i$ that can be found in closed-form as: 
\begin{align}\label{eq:Dj}
D^i=\textstyle{\mathrm{diag}}(\max(\textstyle{\mathrm{sgn}}(\yong{\overline{J}^{\mu}_i}),\mathbf{0}_{1,{n_z}})).
\end{align}
Furthermore, for any interval domain $[\underline{z} \ \overline{z}]\subseteq \mathcal{Z}$, the following inequality holds:
\begin{align}\label{eq:mm_bounding}
\mu_d(\overline{z},\underline{z})-\mu_d(\underline{z},\overline{z})\leq F_{\mu}(\overline{z}-\underline{z}), 
\end{align}
where $\mu_d=[\mu_{d,1},\dots, \mu_{d,p}]^\top$ and 
\begin{align}\label{eq:F_bounding}
F_{\mu}=\overline{J}^\oplus_{\mu}+\underline{J}^\ominus_{\mu}.
\end{align}
Finally, if $\mu$ is computed through the mixed-monotone decomposition in Proposition \ref{prop:JSS_decomp}, then
\begin{align}\label{eq:JSS_Jacobian}
\overline{J}^\mu=\overline{J}^f-H \ \text{and} \ \underline{J}^\mu=\underline{J}^f-H.
\end{align}  
\end{prop}
\section{Problem Formulation} \label{sec:Problem}
\noindent\textbf{\emph{System Assumptions.}} 
Consider the following nonlinear discrete-time (DT) system:  
\begin{align} \label{eq:system}
\begin{array}{ll}
\mathcal{G}: \begin{cases} {x}_{k+1} &= f(x_k)+Bu_k+Ww_k    \\
                                              y_k &= g(x_k)+Du_k+Vv_k
                                              \end{cases},
\end{array}\hspace{-0.2cm}
\end{align}
where $x_k \in \mathcal{X} \subset \mathbb{R}^n$, $u_k \in \mathbb{R}^m$, $y_k \in \mathbb{R}^l$, $w_k \in \mathcal{W} \triangleq [\underline{w},\overline{w}] \subset \mathbb{R}^{n_w}$ and $v_k \in \mathcal{V} \triangleq [\underline{v},\overline{v}] \subset \mathbb{R}^{n_v}$ are continuous state, control input, output (measurement), process disturbance and measurement noise signals, respectively. Furthermore, ${f}:\mathcal{X}  \to \mathbb{R}^n$ and ${g}:\mathcal{X}  \to \mathbb{R}^l$  are known, nonlinear vector fields, and $B \in \mathbb{R}^{n} \times \mathbb{R}^{n}, W \in \mathbb{R}^{n} \times\mathbb{R}^{n_w},D \in \mathbb{R}^{l} \times\mathbb{R}^{b}$ and $V \in \mathbb{R}^{l} \times\mathbb{R}^{n_v}$ are known matrices. 
\begin{assumption}\label{ass:mixed_monotonicity}
 The mappings $f$ and $g$ are known, differentiable and have bounded Jacobians, i.e., satisfy
 \begin{align*}
 \underline{J}_{f} \leq J_{\nu}(x) \leq \overline{J}_{\nu}, \ \forall \nu \in \{f,g\}, \ \forall x \in \mathcal{X},
 \end{align*}
 where $J_f$ and $J_g$ are the Jacobian matrix functions and $\overline{J}_{f},\underline{J}_{f} \in \mathbb{R}^{n \times n}, \overline{J}_{g},\underline{J}_{g} \in \mathbb{R}^{l \times n}$ are  priori known matrices.
\end{assumption}
Based on Assumption \ref{ass:mixed_monotonicity} and by leveraging Proposition \ref{prop:JSS_decomp}, the matrices $A \in \mathbb{R}^{n \times n}$ and $C \in \mathbb{R}^{l \times n}$ are chosen such that the following decompositions hold (cf. Definition \ref{defn:JSS} for JSS mappings):
\begin{align} \label{eq:JSS_decom}
\hspace{-.2cm}\forall x \in \mathcal{X}: \begin{cases}f(x)=Ax+\phi(x) \\ h(x)=Cx+\psi(x) \end{cases} \hspace{-.4cm} s.t. \ \phi,\psi \ \text{are JSS in} \ \mathcal{X}.
\end{align} 
Recall that according to Proposition \ref{prop:tight_decomp}, the JSS mappings $\phi$ and $\psi$ admit matrices $F_{\phi}$ and $F_{\psi}$ that can be computed via \eqref{eq:F_bounding} and satisfy \eqref{eq:mm_bounding}. Without loss of generality, we assume $F_{\phi}$ is invertible. The reason is that by construction, $F_{\phi}$ is a non-negative matrix (cf. \eqref{eq:F_bounding}). So, even if the initially computed $F_{\phi}$ is not invertible, one can increase its diagonal elements until the modified $F_{\phi}$ becomes a diagonally dominant, and so invertible matrix~\cite{sootla2017block}. Furthermore, since the modified $F_{\phi}$ is not less than the initial one, it still satisifes \eqref{eq:mm_bounding} given the non-negativity of all items in both sides of the inequality.

We also assume the following.
\begin{assumption} \label{ass:initial_interval}
 The initial state $x_0$ satisfies $x_0 \in \mathcal{X}_0 = [ \underline{x}_0,\overline{x}_0]$, where $\underline{x}_0$ and $\overline{x}_0$ {are} known initial state bounds, and the values of the output/measurement $y_k$ signals are known at all times.  
 \end{assumption} 
 Our goal is to design the control input signal $u_k$ as a (implicit) function of system measurements/outputs to stabilize the closed-loop DT nonlinear system \eqref{eq:system}--\eqref{eq:JSS_decom}. Due to the existing uncertainties, we do not assume that the state is measurable (even the initial state is uncertain), and we only measure the output $y_k$.
\begin{problem}\label{eq:stab_synth}
Given the nonlinear system in \eqref{eq:system}--\eqref{eq:JSS_decom}, as well as Assumptions \ref{ass:mixed_monotonicity} and \ref{ass:initial_interval}, 
synthesize a stabilizing feedback control input $u_k$ for the plant $\mathcal{G}$ so that the resulting closed-loop system is asymptotically stable. 
\end{problem}
\section{Proposed Control Design}
Before discussing our control design strategy, we first briefly introduce the notions of framer and interval observer which we extensively use throughout the paper.
\begin{defn}[Correct Interval \mk{Framers}]\label{defn:framers}
Given the nonlinear plant $\mathcal{G}$, 
the sequences of signals $\{\overline{x}_k,\underline{x}{k}\}_{k=0}^{\infty}$ in $\mathbb{R}^n$ are called upper and lower framers for the states of \eqref{eq:system}, if 
\begin{align}\label{eq:correctness}
\forall k \in [0,\infty), \ \underline{x}_k \leq x_k \leq \overline{x}_k.
\end{align}
In other words, starting from the initial interval $\underline{x}_0 \leq x_0 \leq \overline{x}_0$, the true state of the system in \eqref{eq:system}, $x_k$, is guaranteed to evolve within the interval flow-pipe $[\underline{x}_k,\overline{x}_k]$, for all $k \geq 0$. Finally, any dynamical system $\hat{\mathcal{G}}$ whose states are correct framers for the states of the plant $\mathcal{G}$, i.e., any (tractable) algorithm that returns upper and lower framers for the states of plant $\mathcal{G}$ is called a \emph{correct} interval \mk{framer} system for \eqref{eq:system}. 
\end{defn}
\begin{defn}[\mk{Framer} Error]\label{defn:error}
Given 
state framers \mk{$\underline{x}_k \leq \overline{x}_k, k \geq 0$}, $\varepsilon_k \triangleq \overline{x}_k-\underline{x}_k$, \mk{whose infinite norm denotes} the interval width \mk{of $[\underline{x}_k,\overline{x}_k]$ (cf. Definition \ref{defn:interval})},  
 is called the \mk{framer} error. It can be easily verified that 
correctness (cf. Definition \ref{defn:framers}) implies that $\varepsilon_k \geq 0, \forall k \geq 0.$  
\end{defn}
\begin{defn}[Stability and \mk{Interval Observer}]\label{defn:stability}
A corresponding interval \mk{framer} to system \eqref{eq:system} is input-to-state stable (ISS), if the \mk{framer} error sequence $\{\varepsilon_{k}\}_{k=0}^{\infty}$ is bounded as follows: 
\begin{align}
\|\varepsilon_{k}\|_2 \leq \beta(\|\varepsilon_0\|,k)+\rho(\|\delta\|_{\ell_{\infty}}), \forall k \geq 0,
\end{align}
where $\delta \triangleq [\delta^\top_w \delta^\top_v]^\top$, $\beta$ and $\rho$ are functions of classes $\mathcal{KL}$ and $\mathcal{K}_{\infty}$, respectively, and
$\|\delta\|_{\ell_{\infty}}=\sup_{k \geq 0} \|\delta_k\|_2=\|\delta\|_2$ is the $\ell_{\infty}$ signal norm.
An ISS interval framer is called an interval observer.
\end{defn}
Our proposed control strategy has the following steps:
\begin{enumerate}[(i)]
    \item Since the actual state of the system in \eqref{eq:system}--\eqref{eq:JSS_decom} is uncertain and not measured, we first construct a corresponding certain and robust \emph{framer system}, that given the output (measurement) signal and by construction, returns upper and lower interval-valued estimates of the state of the actual system, given \emph{any} control signal $u_k$, any to-be-designed observer gain, and for any realization of the bounded noise and disturbance signals.\label{item:framer}
    \item Then, we consider a state feedback control for the framer system constructed in \eqref{item:framer} using to-be-designed control gains. The goal is to stabilize the framer system which results in a stable actual system.\label{item:control}
    \item Finally, by plugging the controls from \eqref{item:control} into the framer system in \eqref{item:framer}, we obtain the closed-loop framer system. We stabilize the closed-loop system by designing the observer and controller gains, which we show that can be done separately. This results in an input-to-state stable closed-loop system for the actual plant.
\end{enumerate}


\subsection{Interval Framers} \label{sec:obsv}
Given the nonlinear plant $\mathcal{G}$, in order to address Problem \ref{eq:stab_synth}, we first propose an interval observer (cf. Definition \ref{defn:framers}) for $\mathcal{G}$ similar to \cite{9790824}, through the following dynamical system:
\begin{align}\label{eq:observer}
\begin{array}{rl}
\overline{x}_{k+1}\hspace{-.1cm}&=(A\hspace{-.1cm}-\hspace{-.1cm}LC)^\oplus \overline{x}_k-(A\hspace{-.1cm}-\hspace{-.1cm}LC)^\ominus \underline{x}_k\hspace{-.1cm}+(LV)^\ominus \overline{v}\hspace{-.1cm}-\hspace{-.1cm}(LV)^\oplus \underline{v} \\
&+\phi_d(\overline{x}_k,\underline{x}_k)+\hspace{-.1cm}Ly_k\hspace{-.1cm}+\hspace{-.1cm}(B\hspace{-.1cm}-\hspace{-.1cm}LD)u_k+W^{\oplus}\overline{w}-W^{\ominus}\underline{w}\\
&+L^{\ominus}\psi_d(\overline{x}_k,\underline{x}_k)-L^{\oplus}\psi_d(\underline{x}_k,\overline{x}_k)\\
\underline{x}_{k+1}\hspace{-.1cm}&=(A\hspace{-.1cm}-\hspace{-.1cm}LC)^\oplus \underline{x}_k-(A\hspace{-.1cm}-\hspace{-.1cm}LC)^\ominus \overline{x}_k\hspace{-.1cm}+(LV)^\ominus \underline{v}\hspace{-.1cm}-\hspace{-.1cm}(LV)^\oplus \overline{v} \\
&+L^{\ominus}\psi_d(\underline{x}_k,\overline{x}_k)-L^{\oplus}\psi_d(\overline{x}_k,\underline{x}_k)\\
&+\phi_d(\underline{x}_k,\overline{x}_k)+\hspace{-.1cm}Ly_k\hspace{-.1cm}+\hspace{-.1cm}(B\hspace{-.1cm}-\hspace{-.1cm}LD)u_k+W^{\oplus}\underline{w}-W^{\ominus}\overline{w}
 \end{array}
\end{align}
where $\phi_d:\mathbb{R}^n \times \mathbb{R}^n \to \mathbb{R}^n$ and $\psi_d:\mathbb{R}^n \times \mathbb{R}^n \to \mathbb{R}^l$ are tight mixed-monotone decomposition functions of the JSS mappings $\phi$ and $\psi$, computed via Proposition. Moreover, $L \in \mathbb{R}^{n \times l}$ is a to-be-designed observer gain matrix. 
Defining $\varepsilon_k \triangleq \overline{x}_k-\underline{x}_k$, we obtain the following system from \eqref{eq:observer} that governs 
 the dynamics of the framer errors:
\begin{align}\label{eq:farmer_error}
 \varepsilon_{k+1}\hspace{-.1cm}=|A-LC|\varepsilon_k+\delta^{\phi}_k+\hspace{-.1cm}|L|\delta^{\psi}_k+|LV|\delta_v+|W|\delta_w,   
\end{align}
where 
\begin{align}\label{eq:deltas}
\begin{array}{rl}
&\overline{\delta}^{\rho}_k \triangleq \rho_d(\overline{x}_k,\underline{x}_k)-\rho(x_k),\underline{\delta}^{\rho}_k \triangleq \rho(x_k)\hspace{-.1cm}-\hspace{-.1cm}\rho_d(\underline{x}_k,\overline{x}_k),\\
& \delta_{\alpha} \triangleq \overline{\alpha}_k-\underline{\alpha}_k, \forall \rho \in \{\phi,\psi\}, \forall \alpha \in \{w,v\}.
\end{array}
\end{align}
The following proposition ensures that \eqref{eq:farmer_error} is an interval framer for \eqref{eq:system}--\eqref{eq:JSS_decom}.
\begin{prop}\label{prop:frame}
The system in \eqref{eq:observer} constructs a framer system for the plant in \eqref{eq:system}--\eqref{eq:JSS_decom} for all values of the control signal $u_k$, all realizations of the bounded noise and disturbance $v_k,w_k$ and all values of the gain $L$. Consequently, \eqref{eq:farmer_error} is a positive system. Furthermore, if $L$ is a solution to the SDP in \cite[(17) and (19)]{9790824}, then \eqref{eq:farmer_error} is an ISS system, i.e., \eqref{eq:observer} is an interval observer for \eqref{eq:system}--\eqref{eq:JSS_decom}.
\end{prop}
\begin{proof}  
The proof follows the lines of~\cite[Theorems 1 \& 2]{9790824} with the slight modification that the terms $Bu_k$ and $Du_k$ are added to the state and output equations, respectively, as well as to the corresponding framers, since they are treated as known variables when upper and lower framers are computed.
\end{proof}
\subsection{Closed-Loop System and Separation Principle}\vspace{-0.05cm}
In this subsection we propose the following nonlinear state feedback control to stabilize the system in \eqref{eq:observer}: 
\begin{align}\label{eq:control}
  u_k=\overline{K}\overline{x}_k-\underline{K}\underline{x}_k+K_{\nu}(\phi_d(\overline{x}_k,\underline{x}_k)-\phi_d(\overline{x}_k,\underline{x}_k)), 
\end{align}
where $\overline{K},\underline{K}$ and $K_{\nu}$ are to-be-designed control gain matrices in $\mathbb{R}^{m \times n}$. The goal is to stabilize the framer system \eqref{eq:observer} which consequently stabilizes the trajectory of the actual closed-loop system~\cite{smith2008global,chu1998mixed}.

It is worth emphasizing the additional degree of freedom of our nonlinear control approach, i.e, $K_{\nu}$ compared to the linear control approach in \cite{efimov2013control}. This extra to-be-designed gain, compensates for the potential dynamic instabilities that are due to the system nonlinearities and extends the feasible set of admissible control signals, hence potentially resulting in tighter closed-loop intervals.

Next, to obtain the dynamics of the closed-loop system, we plug $u_k$ from \eqref{eq:control} into \eqref{eq:observer} and use the fact that 
\begin{align}
\begin{array}{c}
Ly_k+(B-LD)u_k=Bu_k+L(y_k-Du_k)\\
=Bu_k+Cx_k+\psi(x_k)+Vv_k,
\end{array}
\end{align}
which results in the following closed-loop system: 
\begin{align} \label{eq:eqiv_sys}
\begin{array}{rl}
\overline{x}_{k+1}\hspace{-.1cm}&=((A\hspace{-.1cm}-\hspace{-.1cm}LC)^\oplus+B\overline{K}) \overline{x}_k-((A\hspace{-.1cm}-\hspace{-.1cm}LC)^\ominus+B\underline{K}) \underline{x}_k\\
&+(LV)^\ominus \overline{v}\hspace{-.1cm}-\hspace{-.1cm}(LV)^\oplus \underline{v}+LVv_k+W^{\oplus}\overline{w}-W^{\ominus}\underline{w}  \\
&+\phi_d(\overline{x}_k,\underline{x}_k)+L^{\oplus}\overline{\delta}^{\psi}_k+L^{\ominus}\underline{\delta}^{\psi}_k+BK_{\nu}\delta^{\phi}_k+LCx_k,\\
\underline{x}_{k+1}\hspace{-.1cm}&=((A\hspace{-.1cm}-\hspace{-.1cm}LC)^\oplus-B\underline{K}) \underline{x}_k-((A\hspace{-.1cm}-\hspace{-.1cm}LC)^\ominus-B\overline{K}) \overline{x}_k\\
&+(LV)^\ominus \underline{v}\hspace{-.1cm}-\hspace{-.1cm}(LV)^\oplus \overline{v}+LVv_k+W^{\oplus}\underline{w}-W^{\ominus}\overline{w}  \\
&+\phi_d(\underline{x}_k,\overline{x}_k)-L^{\oplus}\overline{\delta}^{\psi}_k-L^{\ominus}\underline{\delta}^{\psi}_k+BK_{\nu}\delta^{\phi}_k+LCx_k,
 \end{array}
\end{align}
where $\overline{\delta}^{\phi}_k$ and $\overline{\delta}^{\psi}_k$ are defined in \eqref{eq:deltas}.

The system in \eqref{eq:eqiv_sys} still contains the actual state $x_k$ which is not measured and hence is not accessible. To resolve this issue, we apply a change if variabes and define the following sequences of \emph{upper and lower closed-loop errors} for $k \geq 0$: 
\begin{align}\label{eq:frame_error}
\begin{array}{rl}
\overline{e}_k &\triangleq \overline{x}_k-x_k \Rightarrow x_k=\overline{x}_k-\overline{e}_k, \\
\underline{e}_k &\triangleq x_k-\underline{x}_k \Rightarrow x_k=\underline{x}_k+\underline{e}_k.
\end{array}
\end{align}
Next, from \eqref{eq:system} and the fact that $$L(y_k-Cx_k-\psi(x_k)-Du_k-Vv_k)=0, \ \forall k,$$ we obtain:
\begin{align*}
x_{k+1}&=(A-LC)x_k-L\psi(x_k)+\phi(x_k)\\
&+Ww_k+Ly_k-LVv_k+(B-LD)u_k.
\end{align*}
Finally, combining this, \eqref{eq:eqiv_sys}, and the fact that $M=M^{\oplus}-M^{\ominus}$ for any matrix $M$, we obtain the dynamics of the closed-loop errors: 
\begin{align}\label{eq:err_sys}
\begin{array}{rl}
\hspace{-.2cm}\overline{e}_{k+1}\hspace{-.1cm}&=(A\hspace{-.1cm}-\hspace{-.1cm}LC)^\oplus \overline{e}_k-(A\hspace{-.1cm}-\hspace{-.1cm}LC)^\ominus \underline{e}_k+\overline{\delta}^{\phi}_k-Ww_k\\
&+(LV)^\ominus \overline{v}\hspace{-.1cm}-\hspace{-.1cm}(LV)^\oplus \underline{v}+LVv_k+W^{\oplus}\overline{w}-\hspace{-.1cm}W^{\ominus}\underline{w},\\
\hspace{-.2cm}\underline{e}_{k+1}\hspace{-.1cm}&=(A\hspace{-.1cm}-\hspace{-.1cm}LC)^\oplus \underline{e}_k-(A\hspace{-.1cm}-\hspace{-.1cm}LC)^\ominus \overline{e}_k+\underline{\delta}^{\phi}_k+Ww_k\\
&+(LV)^\ominus \underline{v}\hspace{-.1cm}-\hspace{-.1cm}(LV)^\oplus \overline{v}+LVv_k+W^{\oplus}\underline{w}-\hspace{-.1cm}W^{\ominus}\overline{w},
 \end{array}
\end{align}   
We are ready to sate our fist result, which shows that the augmented system of closed-loop framers and errors has a comparison system that is independent of $x_k$, and can be stabilized by separately designing the observer gain $L$, and the control gain matrices $\overline{K}, \underline{K}, K_{\nu}$. 
\begin{lem}[Separation Principle]\label{lem:separation}
The augmentation of \eqref{eq:eqiv_sys} and \eqref{eq:err_sys} can be stabilized by first designing a stabilizing observer gain $L$ for the framer error system \eqref{eq:farmer_error}, and then synthesizing the state feedback control gains $\overline{K},\underline{K}, K_{\nu}$, given $L$ to stabilize the augmented system.
\end{lem}
\begin{proof}
Augmenting \eqref{eq:eqiv_sys} and \eqref{eq:err_sys}, and using the following inequalities for $\rho \in \{\phi,\psi\}$ according to Proposition \ref{prop:tight_decomp} and the fact that $\rho_d(\underline{x}_k,\overline{x}_k)\leq \rho(x_k)\leq \rho_d(\overline{x}_k,\underline{x}_k)$, and hence:
\begin{align}\label{eq:F}
\begin{array}{rl}
\delta^{\rho}_k&=\hspace{-.1cm}\rho_d(\overline{x}_k,\underline{x}_k)-\rho(x_k)\leq \rho_d(\overline{x}_k,\underline{x}_k)-\hspace{-.1cm}\rho_d(\underline{x}_k,\overline{x}_k)\\
&\leq F_{\rho}(\overline{x}_k-\underline{x}_k)=F_{\rho}(\overline{x}_k-x_k+x_k-\underline{x}_k)\\
&=F_{\rho}(\overline{e}_k+\underline{e}_k),
\end{array}
\end{align}
returns a comparison augmented system:
\begin{align}\label{eq:compare_sys}
z_{k+1}\leq \tilde{A}z_k+\lambda(z_k)+\Lambda \eta_k, 
\end{align}
where $\eta_k \triangleq [\overline{w}^\top \underline{w}^\top w^\top_k \overline{v}^\top \underline{v}^\top v_k^\top]^\top$,
\begin{align*}
z&\triangleq [\overline{e}^\top \underline{e}^\top \overline{x}^\top \underline{x}^\top]^\top, \lambda(z) \triangleq [\mathbf{0}^\top_n \hspace{.1cm} \mathbf{0}^\top_n \hspace{.1cm} \phi^\top_d(\overline{x},\underline{x}) \hspace{.1cm} \phi^\top_d(\underline{x},\overline{x})]^\top,\\
 \Lambda &\triangleq \begin{bmatrix} W^\oplus & -W^\ominus &-W & (LV)^\ominus & -(LV)^\oplus & LV \\  -W^\ominus & W^\oplus & W & -(LV)^\oplus & (LV)^\ominus & LV \\
 W^\oplus & -W^\ominus &\mathbf{0}_{n,n_w} & (LV)^\ominus & -(LV)^\oplus & LV \\  -W^\ominus & W^\oplus & \mathbf{0}_{n,n_w} & -(LV)^\oplus & (LV)^\ominus & LV
 \end{bmatrix},\\
\tilde{A}&=\begin{bmatrix} [\tilde{A}_{11}] & [\tilde{A}_{12}] & \mathbf{0}_{n\times n} & \mathbf{0}_{n\times n} \\  [\tilde{A}_{21}] & [\tilde{A}_{22}] &\mathbf{0}_{n\times n} & \mathbf{0}_{n\times n} \\ [\tilde{A}_{31}] & [\tilde{A}_{32}] & [\tilde{A}_{33}] &[\tilde{A}_{34}] \\ [\tilde{A}_{41}] & [\tilde{A}_{42}] & [\tilde{A}_{43}] &[\tilde{A}_{44} ]
\end{bmatrix}.
\end{align*}
Moreover, the block matrices inside $\tilde{A}$ are as follows:
\begin{align*}
[\tilde{A}_{11}] &\triangleq (A-LC)^\oplus+F_{\phi}+|L|F_{\psi},\\
[\tilde{A}_{12}] &\triangleq (A-LC)^\ominus+F_{\phi}+|L|F_{\psi},\\
[\tilde{A}_{21}] &\triangleq (A-LC)^\ominus+F_{\phi}-|L|F_{\psi},\\
[\tilde{A}_{22}] &\triangleq (A-LC)^\oplus+F_{\phi}-|L|F_{\psi},\\
[\tilde{A}_{31}] &\triangleq -LC+|L|F_{\psi}+BK_{\nu}F_{\phi},\\ [\tilde{A}_{32}] &\triangleq |L|F_{\psi}+BK_{\nu}F_{\phi},\\
[\tilde{A}_{33}] &\triangleq (A-LC)^\oplus+\hspace{-.1cm}B\overline{K}+LC,\\
[\tilde{A}_{34}] &\triangleq -(A-LC)^\ominus\hspace{-.1cm}-\hspace{-.1cm}B\underline{K}, \\ 
[\tilde{A}_{41}] &\triangleq -|L|F_{\psi}+BK_{\nu}F_{\phi},\\ [\tilde{A}_{42}] &\triangleq LC-|L|F_{\psi}+BK_{\nu}F_{\phi},\\
[\tilde{A}_{43]} &\triangleq -(A-LC)^\ominus+B\overline{K},\\ [\tilde{A}_{44}] &\triangleq (A-LC)^\oplus-\hspace{-.1cm}B\underline{K}\hspace{-.1cm}+\hspace{-.1cm}LC. 
\end{align*}
The comparison system \eqref{eq:compare_sys} has a linear component with state matrix $\tilde{A}$, as well as the Lipschitz nonlinear component $\lambda(z)$. Note that $\lambda$ is a locally Lipschitz mapping since $\phi$ has bounded Jacobin by Assumption \ref{ass:mixed_monotonicity} and Proposition \ref{prop:tight_decomp}, and hence is locally Lipschitz. Consequently, $\phi_d$ is Lipschitz by construction (cf. \eqref{eq:JJ_decomp} in Proposition \ref{prop:tight_decomp}). Since, $\tilde{A}$ is a block lower trinagular matrix, its set of eigenvalues is a superset of the set of eigenvalues of the matrix $\tilde{A}_u \triangleq \begin{bmatrix} \tilde{A}_{11} & \tilde{A}_{12} \\ \tilde{A}_{21} & \tilde{A}_{22} \end{bmatrix}$ which only depend on $L$. Moreover, the nonlinear vector function $\lambda(z)$ does not affect the stability of $\tilde{A}_u$ due to its zero upper elements. Hence, the observer and controller gains can be designed separately to stabilize system \eqref{eq:compare_sys}. Moreover, by applying the last inequality in \eqref{eq:F}, it is straightforward to see that \eqref{eq:err_sys} admits a linear comparison system with the sate matrix $\tilde{A}_u$. Hence, any observer gain that stabilizes \eqref{eq:compare_sys} should stabilize \eqref{eq:err_sys} and vice versa. Finally, from the definitions of $\varepsilon_k, \overline{e}_k, \underline{e}_k$, we have $\mathbf{0}_n \leq \overline{e}_k \leq \varepsilon_k$ and $\mathbf{0}_n \leq \underline{e}_k \leq \varepsilon_k$. Therefore, any observer gain $L$ that stabilizes the open loop framer error system \eqref{eq:farmer_error}, should stabilize \eqref{eq:err_sys}, and consequently \eqref{eq:compare_sys}.      
\end{proof}
\subsection{Control Design}
Based on Lemma \eqref{lem:separation}, the design of the observer and control gains can be done separately, where $L$ can be designed first to stabilize \eqref{eq:farmer_error}. In our previous work~\cite{9790824} (also summarized in Proposition \ref{prop:frame}) we provided linear matrix inequalities, through which such an $L$ can be computed. 
Given this observer gain, in order to obtain tractable LMIs to synthesize the stabilizing feedback control gains $\overline{K},\underline{K},K_{\nu}$, the main challenge to be addressed is to resolve the bilinearities between decision variables when applying the existing results on the stability of Lipschitz nonlinear DT systems to \eqref{eq:compare_sys}. The following theorem and its proof tackles this challenge by applying specific similarity transformations and change of variables.    
\begin{thm}\label{thm:control_design}
Suppose $L \in \mathbb{R}^{n \times l}$ is computed through Proposition \ref{prop:frame}, $\alpha > 0$ is a chosen real constant (e.g., is a desired decay rate for the error system \eqref{eq:err_sys}) picked by the designer, and $\epsilon=\frac{1}{\alpha \gamma^2}-1$, where $\gamma=\|F_{\phi}\|_{\infty}$ is the Lipshitz constant of the decomposition function $\phi_d$. Let
$(\mu_*,\Gamma_*,Q_*,\Theta_*)$ be a solution to the following SDP:
\begin{align}\label{eq:stabilizing_K}
&\min\limits_{\{\mu > 0, \Gamma \succ \mathbf{0}_{4n \times 4n}, Q \succ \mathbf{0}_{4n \times 4n}, \Theta\}}{\mu}\\
\nonumber &\text{s.t.} \ \begin{bmatrix} \Gamma - Q & Q & Q\hat{A}^\top +\Theta \hat{B}^\top & I_{4n} \\
* & -\alpha I_{4n} & \mathbf{0}_{4n \times 4n} & \mathbf{0}_{4n \times 4n} \\
* & * & -\frac{1}{2}Q & Q \\
* & * & * & Q-2\epsilon \Gamma \end{bmatrix} \prec \mathbf{0}_{16n \times 16n},\\
\nonumber & \quad \ \ \begin{bmatrix}
-\mu I_{\tilde{n}} & \Lambda^\top  & \Lambda^\top  \\
* & -\frac{1}{2}Q & \mathbf{0}_{4n \times 4n} \\
* & * & -\Gamma
\end{bmatrix} \prec \mathbf{0}_{\hat{n} \times \hat{n}}, \ \begin{bmatrix} I_{4n} & Q \\
Q & \Gamma \end{bmatrix} \succeq \mathbf{0}_{8n \times 8n}, 
\end{align}
where $\tilde{n}\triangleq 3(n_w+n_v)$ is the dimension of the augmented noise vector $\eta_k \in \mathbb{R}^{\tilde{n}}$ and $\hat{n} \triangleq \tilde{n}+8n$. Moreover, $\hat{A} \in \mathbb{R}^{4n \times 4n}$ and $\hat{B} \in \mathbb{R}^{4n \times 12n}$ are defined as follows:
\begin{align*}
[\hat{A}_{ij}]&=[\tilde{A}_{ij}],  i \in \{1,2\},  j \in \{1,2,3,4\},\\
[\hat{A}_{31}]&= LC+|L|F_{\psi}, [\hat{A}_{32}]= |L|F_{\psi},\\
[\hat{A}_{33}]&=(A-Lc)^\oplus+LC, [\hat{A}_{34}] = -(A-LC)^\ominus, \\
[\hat{A}_{41}]&=-|L|F_{\psi}, [\hat{A}_{42}]= LC-|L|F_{\psi},\\
[\hat{A}_{43}]&=-(A-LC)^\ominus, [\hat{A}_{44}]=(A-LC)^\oplus +LC,\\
[\hat{B}_{ij}]&=
\begin{cases} \mathbf{0}_{n \times 3m}, \quad \quad \quad \quad \quad \quad i \in \{1,2\},  j \in \{1,2,3,4\},\\
 \begin{bmatrix}\mathbf{0}_{n \times m} & \mathbf{0}_{n \times  m} & B\end{bmatrix}, \quad  i \in \{3,4\},  j \in \{1,2\},\\
 \begin{bmatrix}B & \mathbf{0}_{n \times m} & \mathbf{0}_{n \times m}\end{bmatrix},  \quad i \in \{3,4\}, j=3,\\
 \begin{bmatrix}\mathbf{0}_{n \times m} & -B & \mathbf{0}_{n \times m}\end{bmatrix},  i \in \{3,4\}, j=4
\end{cases}\hspace{-.5cm}.
\end{align*}
Then, the augmented closed-loop system \eqref{eq:compare_sys} is ISS and satisfies 
\begin{align}\label{eq:noise-attenuation}
\|z_k\|^2_2 \leq \mu_*\|\tilde{w}_k\|^2_2, \ \forall k,
\end{align}
with the control gains $\overline{K}_*,\underline{K}_*,K_{\nu*}$ that are obtained as follows: 
\begin{align}\label{eq:control_design}
\tilde{K}_*=(Q_*^{-1}\Theta_*)^\top=\tdiag_4(\tdiag(\overline{K}_*, \underline{K}_*,{K}_{\nu*}F_{\phi})).
\end{align}
\end{thm}
\begin{proof}
First, note that the augmented system in \eqref{eq:compare_sys} has a Lipschitz nonlinear component $\lambda$. This is because the vector field $\phi$ has bounded Jacobians by Assumption \ref{ass:mixed_monotonicity} and \eqref{eq:JSS_Jacobian}, and so is Lipschitz. Hence, $\phi_d$ is also Lipschitz by construction (cf. \eqref{eq:JJ_decomp} and \eqref{eq:Dj} in Proposition \ref{prop:tight_decomp}), with the Lipschitz constant $\gamma=\|F_{\phi}\|_{\infty}$. Then, by \cite[Lemma 3]{abbaszadeh2009lmi}, \eqref{eq:compare_sys} is input-to-state stable and satisfies\eqref{eq:noise-attenuation} if the following LMIs hold (which are the special cases of \cite[(15) and (16)]{abbaszadeh2009lmi} with $H$ being an identity matrix, while $N$ and $M$ are set to be zero matrices):
\begin{align}\label{eq:LMI_1}
&\begin{bmatrix}
I_{4n}-P & I_{4n} & \tilde{A}^\top P & \mathbf{0}_{4n \times 4n} \\
* & -\alpha I_{4n} & \mathbf{0}_{4n \times 4n} & \mathbf{0}_{4n \times 4n} \\
* & * & -\frac{1}{2}P & P \\
* & * & * & P-2\epsilon I
\end{bmatrix} \prec \mathbf{0}_{16n \times 16n}, \\
\label{eq:LMI_2}&\begin{bmatrix}
-\mu I_{\tilde{n}} & \Lambda^\top P & \Lambda^\top P \\
* & -\frac{1}{2}P & \mathbf{0}_{4n \times 4n} \\
* & * & -I_{4n}
\end{bmatrix} \prec \mathbf{0}_{\hat{n} \times \hat{n}}.
\end{align}
\eqref{eq:LMI_1} and \eqref{eq:LMI_2} cannot be directly used to design the control gains due to the existence of the terms $P\tilde{A}$ which contains bilinearities in the form of pre- and post- multiplication of matrix $B$ with decision variables, i.e., the control gains. To overcome this difficulty, we apply two similarity transformations by pre- and post-multiplying the LMIs in \eqref{eq:LMI_1} and \eqref{eq:LMI_2} by $\text{diag}(P^{-1},I_{4n},P^{-1},P^{-1})$ and $\text{diag}(I_{\tilde{n}},P^{-1},P^{-1})$, respectively. This, combined with applying the change of variables $Q=P^{-1}$, and defining the new matrix variable $\Gamma$ that  satisfies $\Gamma \succeq P^{-2}=Q^2 \Leftrightarrow \begin{bmatrix} I_{4n} & Q \\
Q & \Gamma \end{bmatrix} \succeq \mathbf{0}_{8n \times 8n}$ by Schur complements, returns the results.  
\end{proof}
\section{Illustrative Example}
The effectiveness of our control design is illustrated by the following example.
Consider a slightly modified version of the system in \cite[Example 2]{abbaszadeh2009lmi} in the form of \eqref{eq:system}--\eqref{eq:JSS_decom} with:
\begin{align*}
\label{eq:exampletwo}   
A&=\begin{bmatrix}
   0.5000 & -0.5975 & 0.3735 & 0.0457 & 0.3575\\
   0.2500 & 0.3000 & 0.4017  & 0.1114 & 0.0227\\
   0.4880 & 01384  & 0.2500  & 0.7500 & 0.7500\\
   0.3838 & 0.0974 & 0.5000  & 0.2500 & 0.5000\\
   0.0347 & 0.1865 & -0.2500 & 0.5000 & 0.2500
   \end{bmatrix},\\
   \phi(x)&=\begin{bmatrix}
   0.1(\sin(x_3)-x_3)\\
   0.2(\sin(x_4)-x_4)\\
   0.3(\sin(x_1)-x_1)\\
   0\\
   0.1(\sin(x_2)-x_2)
   \end{bmatrix},
   B=\begin{bmatrix}
   0.7 & 0.8 & 0\\
   0.4 & 0.9 & 0.9\\
   0.9 & 0.9 & 0.2\\
   0.9 & 0.6 & 0.7\\
   0   & 0.5 & 0.3
   \end{bmatrix},\\
   D&=\begin{bmatrix}
   0.2 & 0.1 & 0\\
   0.2 & 0.1 & 0
   \end{bmatrix},
   C=\begin{bmatrix}
   0.5 & 0.2 & 0 & 0 & 0.3\\
   0   & 0.2 & 0.1 & 0.3 & 0
   \end{bmatrix}, V=I_2\\
   \psi(x)&=[0.1(\cos(x_1)-x_1) \ 0.2(\cos(x_2)+x_2)]^\top, W=I_5,
\end{align*}
where 
$\mathcal{W}=[-.1 \ .1]^5,\mathcal{V}=[-.1 \ .1]^2$ and $\mathcal{X}_0 = [-6 \ 6]^5$. The nonlinear system is unstable. Moreover, matrix $A$ is unstable, which means the linear part of the system is itself unstable. The first plot in Figure \ref{fig:states} (i.e., in the upper left corner) shows the trajectory of all the states for the unstable open-loop system, where some of the state values grow unbounded as expected. Using Proposition \ref{prop:tight_decomp} we computed $F_{\phi}=\begin{bmatrix} \varepsilon_0 & 0 & 0.2 & 0 & 0  \\ 0 & \varepsilon_0 & 0.4 & 0 & 0\\0.6 &  0 & \varepsilon_0 & 0 & 0 \\ 0 & 0 & 0 & \varepsilon_0 & 0\\ 0 & 0.2 & 0 & 0 & \varepsilon_0  \end{bmatrix}, F_{\psi}=\begin{bmatrix} 0.2 & 0 \\ 0 & 0.4 \end{bmatrix}, \gamma=\|F_{\phi}\|_{\infty}=0.6010$, where $\varepsilon_0=0.001$ is a small number added to the diagonal entries of the original $F_{\phi}$ to make it an invertible matrix. We set $\alpha=0.1$ which results in $\epsilon =\frac{1}{\alpha \gamma^2}-1=26.6854$. Moreover, based on the separation principle results in Lemma \ref{lem:separation}, we fist obtained $L= \begin{bmatrix}
       0.1   &  0   &   0 & 0 & 0\\
    -20.0999 &  1    & 0 & 0 & 0 
    \end{bmatrix}$ 
    by solving the SDP in \cite[(17) and (19)]{9790824} employing YALMIP and MOSEK. Then, given $L$, we leveraged bisection and YALMIP again to solve the corresponding SDP in \eqref{eq:stabilizing_K}, which resulted in the following control gains $\overline{K}_* =-\underline{K}_*=\begin{bmatrix} -0.7686 & 0.6505 & 0 & 0 & 0.1 \\ 
-0.3008 & -1.2088 & 0 & 0 & 0.2 \\ 
0.3207 & -0.6403 & 0 & 0 & 0\end{bmatrix}, {K}_{\nu*} =\begin{bmatrix} 0.3343 & -0.3202 & 0 & 0 & -0.1\\ 
0.1004 & -0.1044 & 0 & 0 & -0.05\\ 
-0.1103 & 0.3201 & 0 & 0 & 0\end{bmatrix}$, as well as the optimal noise attenuation level $\mu_*=1.6137$. 
The second to sixth plots in Figure \ref{fig:states} 
 show the trajectories of the closed-loop stabilized framers as well as actual states, which as can be seen are asymptotically stable. 

\begin{figure}[t!] 
\centering
{\includegraphics[width=0.49\columnwidth]{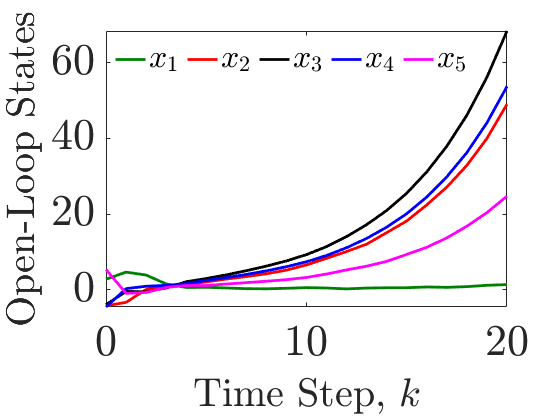}}
{\includegraphics[width=0.49\columnwidth]{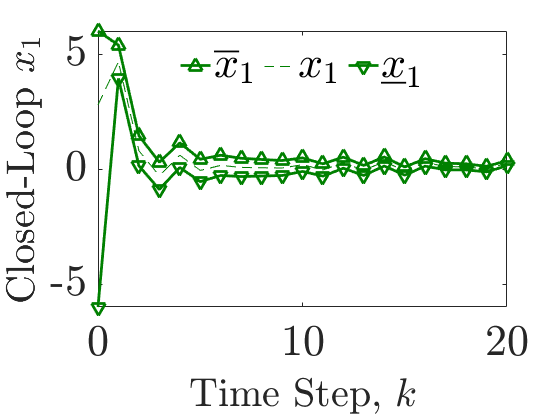}}
{\includegraphics[width=0.49\columnwidth]{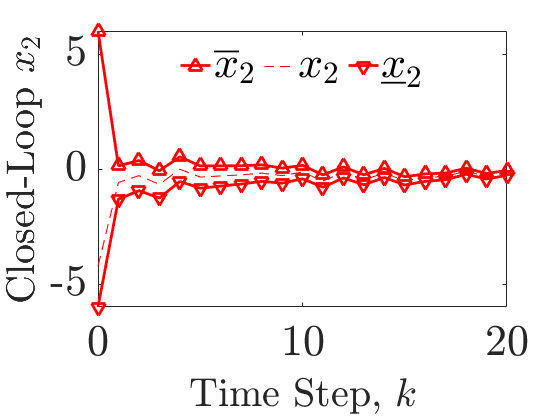}}
{\includegraphics[width=0.49\columnwidth]{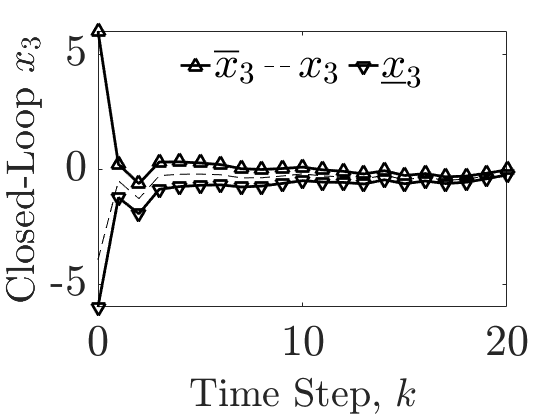}}
{\includegraphics[width=0.49\columnwidth]{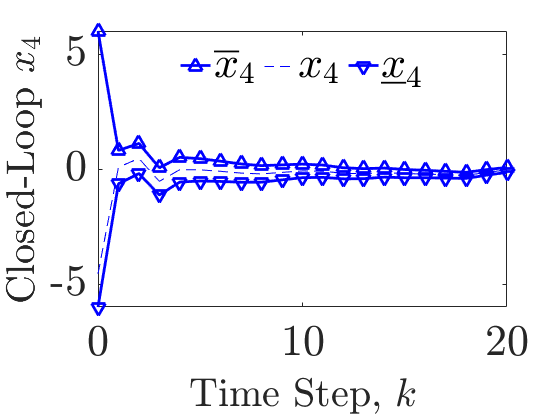}}
{\includegraphics[width=0.49\columnwidth]{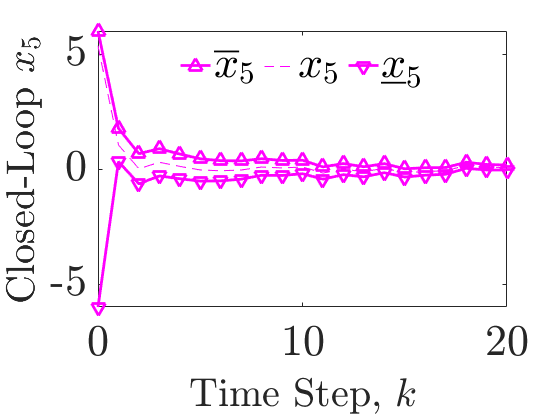}}
\caption{\small\mk{Open-loop states (first plot), as well as the closed-loop upper and lower framers and actual states (second to sixth plots), returned by our proposed control design.}}
\label{fig:states}
\end{figure}
\section{Conclusion and Future Work} \label{sec:conclusion}
In this paper, optimal feedback stabilizing control for bounded Jacobian nonlinear discrete-time systems with nonlinear observations, subject to state and process noise, was addressed. Rather than stabilizing the uncertain system directly, a higher-dimensional interval observer was stabilized, with its states containing the actual system states. The nonlinear control approach provided additional flexibility compared to linear methods, compensating for system nonlinearities and allowing tighter closed-loop intervals. A separation principle was established, enabling the independent design of observer and control gains, and tractable LMIs were derived for control gain design. Extending this approach to switching dynamics and hybrid systems, as well as other control strategies such as model predictive control will be considered for future work.
\bibliographystyle{unsrturl}

\bibliography{biblio}

\end{document}